\documentclass[12pt,a4paper]{article}%
\usepackage{amsmath}
\usepackage{amsfonts}
\usepackage{amssymb}
\usepackage{graphicx}%
\providecommand{\U}[1]{\protect\rule{.1in}{.1in}}
\newtheorem{theorem}{Theorem}
\newtheorem{acknowledgement}[theorem]{Acknowledgement}

\newtheorem{definition}[theorem]{Definition}

\newtheorem{lemma}[theorem]{Lemma}

\newtheorem{proposition}[theorem]{Proposition}
\newtheorem{remark}[theorem]{Remark}

\newenvironment{proof}[1][Proof]{\noindent\textbf{#1.} }{\ \rule{0.5em}{0.5em}}
\begin{document}

\title{Nontrivial Periodic Minimizer for Landau-Brazovskii Model with Constraint}
\author{Yuanlong Ruan}
\maketitle

\begin{abstract}
We study the Landau-Brazovskii model with constraint, which is in the form of
a second order variational problem on the real line. By reducing to handy
situations, we find a nontrivial periodic minimal solution. Moreover, the
proof is kept as simple and self-contained as possible in our specific case.



\end{abstract}

\section{Introduction}

\subsection{A digression on block copolymer}

Block copolymer, a synthesized polymer material, has found many applications
in industry. It is consisting of multiple sequences of monomer alternating in
series with different monomer blocks. The combination of different polymers
endows the polymer material with rich properties, which are the key to their
important applications. An example of such property manipulation can be seen
in poly(urethane) foams, which are used in bedding and upholstery.
Poly(urethane), a multiblock copolymer, is characterized by high-temperature
resilience and low-temperature flexibility. Another important use of block
copolymers is in industrial melt-adhesives. By combining polystyrene with
polymers which exhibit rubber-like and adhesive properties, sturdy adhesives
can be formed which are activated by heat. The structure of this copolymer
utilizes polystyrene blocks on the outside and the rubber block on the inside.
When heat is applied, the polystyrene parts melt and allow for limited
liquid-like flow. The middle section causes adhesion and after the temperature
drops, the strength of polystyrene is restored. This property, made possible
by the combination of polystyrene with other polymers, makes this block
copolymer an important adhesive.

The importance of block copolymer has drawn attention to mathematicians.
Various mathematical models are developed to explore the properties of block
copolymer. Landau-Brazovskii model, among the most popular models, is the
simplest yet desirable in capturing the nature of copolymer and thus gained
much acceptance. It will be our protagonist in this paper.

\subsection{Model specification}

Landau-Brazovskii is formulated as below,%
\begin{equation}
I_{\infty}\left(  \varpi\right)  =\int_{\mathbb{R}}\left\{  \frac{\xi^{2}}%
{2}\left[  \varpi^{\prime\prime}\left(  t\right)  +\varpi\left(  t\right)
\right]  ^{2}+\frac{\tau}{2}\varpi^{2}\left(  t\right)  -\frac{\gamma}%
{6}\varpi^{3}\left(  t\right)  +\frac{1}{24}\varpi^{4}\left(  t\right)
\right\}  dt, \label{model}%
\end{equation}
which models the energy of block copolymer in terms of specific controlling
parameters $\xi>0,$ $\tau,\gamma\in\mathbb{R}$.

Block copolymer is determined by its molecular arrangement. Each molecular
structure or state, by the language of the model, is represented by $\varpi$.
For the purpose of application, the structure or $\varpi$ is required to meet
certain practical restrictions,

\begin{enumerate}
\item[(i)] $\varpi\left(  t\right)  $ is a periodic function with period, say,
$\tau=\tau\left(  \varpi\right)  >0;$

\item[(ii)] The structure $\varpi\left(  t\right)  $ is evenly distributed,%
\[
\int_{0}^{\tau}\varpi\left(  t\right)  dt=0;
\]

\item[(iii)] The block copolymer should behave steadily around the state
$\varpi\left(  t\right)  $, to be precise, structure $\varpi\left(  t\right)
$ should not slide easily into another state when the controlling parameters
are slightly perturbed, therefore, it is necessary for $\varpi\left(
t\right)  $ to have the minimal energy in some sense.
\end{enumerate}

\subsection{Criterion for minimizer}

The model posed in the last subsection comes down to a minimizing problem with
constraints. But it is not quite the minimizing problem which can be solved
straightforwardly using the direct method in the calculus of variation. For
one thing, the energy functional is defined on the real line and the period of
$\varpi\left(  t\right)  $ is allowed to vary, hence the compactness of the
minimizing sequence is thus lossed. On the other hand, the criterion for
minimizer is not clearly, since there are functions for which $I_{\infty
}\left(  \varpi\right)  =-\infty$, in which case a strightforward minimization
does not make sense. Hence the criterion for minimizer needs to be
specifically understood. One way to overcome this difficulty is to minimize
the average energy rather than the energy itself. Specifically, consider the
problem%
\[
\liminf_{T\rightarrow+\infty}\frac{1}{T}I_{T}\left(  f;\varpi\right)
\rightarrow\min,
\]
to be noted, it will be shown in Proposition \ref{THequal} below that the
functional $\left(  \ref{model}\right)  $ on $\mathbb{R}$ may be replaced by
one on the half real line $\mathbb{R}^{+}$. But this optimal criterion is much
two loose for practical purposes, since it will never fail to find functions
which may have different behaviors on compacts but still reach the same
minimal average energy level. Another criterion for infinite horizon problems
founds its source in optimal control problem in economics, it was introduced
by \cite{DG} and \cite{vW}. It is referred to as "overtaking optimality
criterion". As \cite{L&M}, we shall adapt this optimal criterion, such a
minimizer will have minimal energy on every compact intervals and minimal
average energy on the whole real line, these will be made clear in definition
\ref{strong_m}. This specification of minimizer has the advantage that its
mathematical properties are physically desirable.

The Landau-Brazovskii model $\left(  \ref{model}\right)  $ has been employed
by many chemists and physists to simulate the block copolymer. However, few
literatures have been devoted to the exploration of mathematical nature of the
model. In fact, this is a difficult problem and needs in-depth investigation.
\cite{L&M} may be the first effort in this direction, the authors proved
existence of a global periodic minimizer of $\left(  \ref{model}\right)  $
without constaints by localizing. Later \cite{CMM} and \cite{M&Z2} studied the
cosntrained version of model $\left(  \ref{model}\right)  $. Both \cite{L&M}
and \cite{CMM} assume a general functional form. A.J.Zaslavski also extended
the studies of \cite{L&M}\cite{CMM} and made an investigation into the
structure and turpike properties of the optimal solutions, one may refer to
\cite{M&Z2} and references therein. In particular, $\left(  \ref{model}%
\right)  $ is related to a class of fourth order differential equations, for
further reference in this respect, the readers are sent to
\cite{RuanEFK0,RuanEFK1,RuanEFK2,PeletiertTroy}.

In \cite{Ruan}, the model $\left(  \ref{model}\right)  $ with symmetric double
well potential is studied and the existence of global periodic minimizer is
shown. The author also proved the symmetric property of the minimizer. In
addition, without the presence of symmetric property, a non-existence result
is given there. In this paper, we study the constrained version of
\cite{Ruan}'s model, but we do not assume any symmetric properties of $\left(
\ref{model}\right)  $. The proof is given following that of \cite{M&Z2},
however, we are interested in the existence of nontrivial periodic minimizer,
a sufficient condition for the existence is given. Since our model takes a
specific functional form, the proof can be as simple and self-contained as
possible. We shall see that the condition given in \cite{L&M}, which ensures
the existence of nontrivial perodic minimizer, is also the key to the
existence of nontrivial perodic minimizer of the contrained problem.

\subsection{Notations and preparations}

Before going further, we introduce some notions.%

\begin{equation}
f\left(  x,y,z\right)  =\frac{\xi^{2}}{2}z^{2}-\xi^{2}y^{2}+\frac{\xi^{2}%
+\tau}{2}x^{2}-\frac{\gamma}{6}x^{3}+\frac{1}{24}x^{4}, \label{intgrd}%
\end{equation}
the function is determined by controlling parameters $\xi,$ $\tau,$ $\gamma$.
Denote the energy on a bounded interval $\left[  T_{1},T_{2}\right]  $ by%

\[
I_{T_{1},T_{2}}\left(  f;\varpi\right)  =\int_{T_{1}}^{T_{2}}f\left(
\varpi,\varpi^{\prime},\varpi^{\prime\prime}\right)  dt,
\]
Note the integrand is independent of $t$, we have%
\[
I_{T_{1},T_{2}}\left(  f;\varpi\right)  =I_{0,T_{2}-T_{1}}\left(
f;\varpi\right)  .
\]

Therefore, for simplicity of notions, we always use $I_{T}\left(
f;\varpi\right)  $ to represent the integration on any bounded interval of
length $T>0$. \textit{This convention will be adopted as appropriate through
out in this paper}.

Corresponding to $I_{T}\left(  f;\varpi\right)  ,$ we denote by $J_{T}\left(
f;\varpi\right)  $ the average energy on bounded intervals of length as
$T>0,$
\[
J_{T}\left(  f;\varpi\right)  =\frac{1}{T}I_{T}\left(  f;\varpi\right)  .
\]

We have the minimization problem $\mathbb{P}_{T}\left(  f;x,y\right)  ,$%
\[
\zeta_{T}\left(  f;x,y\right)  =\inf\left\{  \left.  J_{T}\left(
f;\varpi\right)  \right\vert \varpi\in\mathcal{A}_{T}\left(  x,y\right)
\right\}  ,
\]
where%
\[
\mathcal{A}_{T}\left(  x,y\right)  =\left\{  \left.  \varpi\in W^{2,1}\left(
0,T\right)  \right\vert \mathcal{V}_{\varpi}\left(  0\right)  =x,\mathcal{V}%
_{\varpi}\left(  T\right)  =y\right\}  ,
\]%
\[
\mathcal{V}_{\varpi}\left(  s\right)  =\left(  \varpi\left(  s\right)
,\varpi^{\prime}\left(  s\right)  \right)  \in\mathbb{R}^{2}.
\]
Note that $\mathcal{V}_{\varpi}\left(  s\right)  $ denotes the vector formed
by the value of $\varpi$ and its derivative at the point $s.$ The notion will
be used frequently.%

\[
E=W_{loc}^{2,1}\left(  \mathbb{R}^{+}\right)  \cap W^{1,\infty}\left(
\mathbb{R}^{+}\right)
\]

Minimization problem $\mathbb{P}^{+}$ on infinite interval $\mathbb{R}^{+}$ is
denoted by,%
\[
\psi_{f}^{+}=\inf\left\{  \liminf\limits_{T\rightarrow\infty}J_{T}\left(
f;\varpi\right)  \left\vert \varpi\in E\right.  \right\}  ,
\]
$\Theta\left(  f\right)  $ is the set of minimizers and $\tilde{\Theta}\left(
f\right)  $ is set of periodic minimizers.

Minimization problem with constraint $\mathbb{P}_{T}\left(  f,a;x,y\right)  $
on finite interval%
\[
\zeta_{T}\left(  f,a;x,y\right)  =\inf\left\{  \left.  J_{T}\left(
f;\varpi\right)  \right\vert \varpi\in\mathcal{A}_{T}\left(  a;x,y\right)
\right\}  ,
\]
where%
\[
\mathcal{A}_{T}\left(  a;x,y\right)  =\left\{  \left.  \varpi\in
W^{2,1}\left(  0,T\right)  \right\vert \left[  \varpi\right]  _{T}%
=a,\mathcal{V}_{\varpi}\left(  0\right)  =x,\mathcal{V}_{\varpi}\left(
T\right)  =y\right\}  ,
\]%
\[
\left[  \varpi\right]  _{T}=\frac{1}{T}\int_{0}^{T}\varpi\left(  t\right)
dt.
\]

Minimization problem on infinite interval $\mathbb{R}^{+}$ with constraint
$\mathbb{P}^{+}\left(  f;a\right)  $ is denoted by%
\begin{equation}
\psi_{f}^{+}\left(  a\right)  =\inf\left\{  \liminf\limits_{T\rightarrow
\infty}J_{T}\left(  f;\varpi\right)  \left\vert \varpi\in E,\left[
\varpi\right]  =a\right.  \right\}  , \label{constr_min}%
\end{equation}
where%
\[
\left[  \varpi\right]  =\liminf\limits_{T\rightarrow\infty}\left[
\varpi\right]  _{T}=\liminf\limits_{T\rightarrow\infty}\frac{1}{T}\int_{0}%
^{T}\varpi\left(  t\right)  dt.
\]
$\Theta\left(  f;a\right)  $ is the set of minimizers and $\tilde{\Theta
}\left(  f;a\right)  $ is set of periodic minimizers.

\begin{remark}
\label{exists}It is shown in \cite{CMM} that, for any $\varpi$ in the domain
of $\mathbb{P}\left(  f;a\right)  ,$ there is $\tilde{\varpi}\in E$ such that
$\lim\limits_{T\rightarrow\infty}J_{T}\left(  f;\tilde{\varpi}\right)
,\lim\limits_{T\rightarrow\infty}\left[  \tilde{\varpi}\right]  _{T}$ exist
and%
\[
\lim\limits_{T\rightarrow\infty}J_{T}\left(  f;\tilde{\varpi}\right)
=\liminf\limits_{T\rightarrow\infty}J_{T}\left(  f;\varpi\right)  ,
\]%
\[
\left[  \tilde{\varpi}\right]  =\lim\limits_{T\rightarrow\infty}\left[
\tilde{\varpi}\right]  _{T}=\left[  \varpi\right]  .
\]
In view of this observation, we may replace $\mathbb{P}\left(  f\right)  $
(respectively $\mathbb{P}\left(  f;a\right)  $) with the following%
\[
\psi_{f}^{+}=\inf\left\{  \lim\limits_{T\rightarrow\infty}J_{T}\left(
f;\varpi\right)  \left\vert \varpi\in E,\text{and }\lim\limits_{T\rightarrow
\infty}J_{T}\left(  f;\varpi\right)  \text{ exists}\right.  \right\}  .
\]
(respectively,%
\[
\psi_{f}^{+}\left(  a\right)  =\inf\left\{  \lim\limits_{T\rightarrow\infty
}J_{T}\left(  f;\varpi\right)  \left\vert \varpi\in E,\text{and }%
\lim\limits_{T\rightarrow\infty}J_{T}\left(  f;\varpi\right)  ,\text{ }\left[
\varpi\right]  \text{ exist, }\left[  \varpi\right]  =a\right.  \right\}  .
\]
)
\end{remark}

\begin{definition}
Assume $\left\{  \varpi_{n}\right\}  _{n=0}^{\infty}$ is a sequence of
functions in $\mathbb{P}\left(  f,a\right)  $ or $\mathbb{P}\left(  f\right)
$, $\left\{  A_{n}\right\}  _{n=0}^{\infty}$ is any sequence of posotive
values and $\left\{  k_{n}\right\}  _{n=0}^{\infty}$ of positive integers,
define%
\[
\alpha_{0}=\frac{A_{0}}{2},\alpha_{n}=k_{n}A_{n}+\alpha_{n-1},n\geqslant0.
\]%
\[
J_{n}=\left[  0,A_{n}\right]  ,n\geqslant0.
\]%
\[
L_{-1}=\left[  -\frac{A_{0}}{2},\frac{A_{0}}{2}\right]  ,L_{n}=\left(
\alpha_{n-1},\alpha_{n}\right]  ,n\geqslant0.
\]%
\[
\hat{\varpi}_{n}\left(  x\right)  =\varpi_{n}\left(  l_{n}\left(
x-\alpha_{n-1}\right)  \right)  ,
\]%
\[
l_{n}\left(  x\right)  =x-\left[  \frac{x}{A_{n}}\right]  A_{n},
\]
where $l_{n}\left(  x\right)  $ maps $\mathbb{R}$ into $J_{n},$ the functions
$\hat{\varpi}$ is called a mixture of $\left(  \left\{  \varpi_{n}\right\}
,\left\{  A_{n}\right\}  ,\left\{  k_{n}\right\}  \right)  $, precisely,%
\[
\hat{\varpi}=Mix\left(  \left\{  \varpi_{n}\right\}  ,\left\{  A_{n}\right\}
,\left\{  k_{n}\right\}  \right)  ,
\]%
\[
\hat{\varpi}\left(  x\right)  =\left\{
\begin{array}
[c]{ll}%
\hat{\varpi}_{n}\left(  x\right)  & x\in L_{n},n\geqslant0,\\
\hat{\varpi}_{0}\left(  x\right)  & x\in L_{-1},\\
\hat{\varpi}_{n}\left(  x+\alpha_{n}+\alpha_{n-1}\right)  & x\in\left(
-\alpha_{n},-\alpha_{n-1}\right]  ,n\geqslant1
\end{array}
\right.  .
\]

\end{definition}

This method of mixture was introduced by \cite{CMM}, the conclusion in remark
\ref{exists} is proved by the method of mixture. Similarly, the method of
mixture is also employed to prove (\cite{M&Z2})%
\begin{equation}
\psi_{f}=\psi_{f}^{+},\text{ }\psi_{f}\left(  a\right)  =\psi_{f}^{+}\left(
a\right)  , \label{equal}%
\end{equation}
where $\psi_{f}^{+},$ $\psi_{f}^{+}\left(  a\right)  $ respectively are the
minimmum for unconstrained and constrained minimization problem on
$\mathbb{R}^{+}$, $\ $and $\psi_{f},$ $\psi_{f}\left(  a\right)  $ are those
on $\mathbb{R}$. Since in \cite{M&Z2}, the proof for \ref{equal} is not
provided, we give it here in detail.

\begin{proposition}
\label{THequal}For all $a\in\mathbb{R}$, $\psi_{f}\left(  a\right)  =\psi
_{f}^{+}\left(  a\right)  $. We also have $\psi_{f}=\psi_{f}^{+}$.
\end{proposition}

\begin{proof}
We only prove the first equality, the other being similar. Let $u\left(
t\right)  $ be an optimal solution to the problem on the whole real line,%
\[
\lim\limits_{T\rightarrow+\infty}\frac{1}{2T}\int_{-T}^{T}f\left(
u,u^{\prime},u^{\prime\prime}\right)  dt=\psi_{f}%
\]
Let $\alpha_{n}$ be an increasing sequence of positive numbers,
\[
\lim\limits_{n\rightarrow+\infty}\frac{1}{2\alpha_{n}}\int_{-T_{n}}%
^{\alpha_{n}}f\left(  u,u^{\prime},u^{\prime\prime}\right)  dt=\psi_{f}.
\]
Denote%
\[
\beta_{n}=\frac{1}{2\alpha_{n}}\int_{-\alpha_{n}}^{\alpha_{n}}\left\vert
f\left(  u,u^{\prime},u^{\prime\prime}\right)  \right\vert dt,
\]
let $k_{n}$ be such that%
\[
\frac{\alpha_{n+1}}{\sum\limits_{l=0}^{n}k_{l}\alpha_{l}}\rightarrow
0,\frac{\beta_{n+1}}{\sum\limits_{l=0}^{n}k_{l}\beta_{l}}\rightarrow0
\]%
\[
A_{m}=2\left(  \alpha_{n}+1\right)
\]%
\[
v_{m}\left(  t\right)  =\left\{
\begin{array}
[c]{cc}%
u\left(  t\right)  & t\in\left[  -\alpha_{n},\alpha_{n}\right] \\
u\left(  t\right)  =u^{\prime}\left(  t\right)  =0 & t=\pm\left(  \alpha
_{n}+1\right)
\end{array}
,\right.
\]
$v_{m}\left(  t\right)  $ is a polynomial of degree $3$ on $\left[
-\alpha_{n}-1,-\alpha_{n}\right)  $ and on $\left(  \alpha_{n},\alpha
_{n}+1\right)  .$%
\[
v=Mix\left(  \left\{  v_{m}\right\}  ,\left\{  A_{m}\right\}  ,\left\{
k_{m}\right\}  \right)  ^{+}.
\]
Then%
\[
\lim\limits_{T\rightarrow+\infty}\frac{1}{T}\int_{0}^{T}f\left(  v,v^{\prime
},v^{\prime\prime}\right)  dt=\psi_{f},
\]
hence%
\[
\psi_{f}\geqslant\psi_{f}^{+}.
\]
For the converse, Let $u\left(  t\right)  $ be an optimal solution to the
problem on the positive real line,%
\[
\lim\limits_{T\rightarrow+\infty}\frac{1}{T}\int_{0}^{T}f\left(  u,u^{\prime
},u^{\prime\prime}\right)  dt=\psi_{f}^{+},
\]
Denote%
\[
\beta_{n}=\frac{1}{2\alpha_{n}}\int_{0}^{2\alpha_{n}}\left\vert f\left(
u,u^{\prime},u^{\prime\prime}\right)  \right\vert dt,
\]
let $k_{n}$ be such that%
\[
\frac{\alpha_{n+1}}{\sum\limits_{l=0}^{n}k_{l}\alpha_{l}}\rightarrow
0,\frac{\beta_{n+1}}{\sum\limits_{l=0}^{n}k_{l}\beta_{l}}\rightarrow0
\]%
\[
A_{m}=2\left(  \alpha_{n}+1\right)
\]%
\[
v_{m}\left(  t\right)  =\left\{
\begin{array}
[c]{cc}%
u\left(  t\right)  & t\in\left[  -\alpha_{n},\alpha_{n}\right] \\
u\left(  t\right)  =u^{\prime}\left(  t\right)  =0 & t=\pm\left(  \alpha
_{n}+1\right)
\end{array}
,\right.
\]
$v_{m}\left(  t\right)  $ is a polynomial of degree $3$ on $\left[
-\alpha_{n}-1,-\alpha_{n}\right)  $ and on $\left(  \alpha_{n},\alpha
_{n}+1\right)  .$%
\[
v=Mix\left(  \left\{  v_{m}\right\}  ,\left\{  A_{m}\right\}  ,\left\{
k_{m}\right\}  \right)  .
\]
Then%
\[
\lim\limits_{T\rightarrow+\infty}\frac{1}{2T}\int_{-T}^{T}f\left(
v,v^{\prime},v^{\prime\prime}\right)  dt=\psi_{f}^{+},
\]
hence%
\[
\psi_{f}\leqslant\psi_{f}^{+}.
\]

\end{proof}

\begin{remark}
In view of this proposition, we need only to consider the problem on the
positive half real line. In the following, we shall abandon the notion
$\psi_{f}^{+}$ (resp. $\psi_{f}^{+}\left(  a\right)  $) and use $\psi_{f}$
(resp. $\psi_{f}\left(  a\right)  $) to indicate the minimizing problems in question.
\end{remark}

\begin{definition}
The differential $\partial\vartheta\left(  x^{\ast}\right)  $ of a convex
function $\vartheta\left(  x\right)  $ at $x=x^{\ast}$ is defined as the set%
\[
\left\{  \lambda\in\mathbb{R}:\vartheta\left(  z\right)  \geqslant
\vartheta\left(  x^{\ast}\right)  +\lambda\left(  z-x^{\ast}\right)  ,\forall
z\in\mathbb{R}\right\}  .
\]
If there exists $\lambda\in\mathbb{R}$ such that,%
\[
\vartheta\left(  z\right)  >\vartheta\left(  \bar{x}\right)  +\lambda\left(
z-\bar{x}\right)  ,\forall z\neq\bar{x},
\]
then $\bar{x}$ is called the exposed point of $\vartheta.$
\end{definition}

\section{Main result and open problem}

We consider only the case of zero mean constraint (i.e., $\left[
\varpi\right]  =0$), other case being similar.

Recall that $e$ is an extreme point of a convex set $\mathcal{K}$, if $e$ does
not belong to the segment (excluding the end points) connecting any two points
$e_{1},e_{2}\in\mathcal{K}$. The extreme point to a convex function is defined
as extreme point to its graph.

\begin{theorem}
\label{RCP}Let $f$ be the energy configuration $\left(  \ref{intgrd}\right)  $
determined by parameters $\xi>0,$ $\tau,\gamma\in\mathbb{R}$ such that $0$ is
an exposed point of $\psi_{f}\left(  x\right)  $, and%
\begin{equation}
\psi_{f}\left(  0\right)  <m_{f}=\inf\left\{  f\left(  t,0,0\right)
\left\vert t\in\mathbb{R}\right.  \right\}  . \label{inequ_rcp}%
\end{equation}
then there is a nontrivial periodic solution to the constrained minimization
problem $\mathbb{P}\left(  f;0\right)  $.
\end{theorem}

Before diving into the proof of the theorem, we would like to put down a few
remarks which we formulated into the following open problems.

\textbf{Open problems}

The model $\left(  \ref{model}\right)  $ has been put forward to help simulate
copolymer, thus it is expected to have certain conditions for the existence of
periodic minimizer that are easy to verify. The result of Theorem \ref{RCP}
may be succinct itself. However, being exposed point is an property that is
difficult to validate both from theoretical and numerical perspective. Our
problem is whether we can find an alternative condition for being an exposed
point. For this, we have the following conjecture.

Let $h\left(  x\right)  $ be any potential function with double well (not
necessarily symmetric), the potential term in $\left(  \ref{model}\right)  $
is an example of such $h\left(  x\right)  $,
\[
h\left(  x\right)  =\frac{\tau}{2}x^{2}-\frac{\gamma}{6}x^{3}+\frac{1}%
{24}x^{4},
\]
Condsider the functional%
\[
I_{T}\left(  \varpi\right)  =\int_{0}^{T}\frac{\xi^{2}}{2}\left[
\varpi^{\prime\prime}\left(  t\right)  +\varpi\left(  t\right)  \right]
^{2}+h\left(  \varpi\left(  t\right)  \right)  dt,
\]
where $\varpi\in E$ is periodic in $t$ (the period is different from $\varpi$
to $\varpi$), recalling%
\[
E=W_{loc}^{2,1}\left(  \mathbb{R}^{+}\right)  \cap W^{1,\infty}\left(
\mathbb{R}^{+}\right)  .
\]
If we write%
\[
h^{\ast}\left(  x\right)  =\frac{\xi^{2}}{2}x^{2}+h\left(  x\right)
\]
and denote by $\bar{h}^{\ast}\left(  x\right)  $ the convex hull of $h^{\ast
}\left(  x\right)  .$Then is it true that
\[
\psi_{f}^{+}\left(  a\right)  =\bar{h}^{\ast}\left(  a\right)  ,
\]
where $\psi_{f}^{+}\left(  a\right)  $ is similarly defined as%
\[
\psi_{f}^{+}\left(  a\right)  =\inf\left\{  \liminf\limits_{T\rightarrow
\infty}\frac{1}{T}I_{T}\left(  \varpi\right)  \left\vert \varpi\in E,\left[
\varpi\right]  =a\right.  \right\}
\]
and%
\[
\left[  \varpi\right]  =\liminf\limits_{T\rightarrow+\infty}\frac{1}{T}%
\int_{0}^{T}\varpi\left(  t\right)  dt.
\]
If the above conjecture is true, then the condition of exposed point can be
dropped. If it is not true, then are there any other alternative conditions
for being an exposed point which are easy to verify ?

\section{Preliminaries}

\begin{theorem}
The problem on the entire real line has the same minimum as that on the harf
real line. Therefore, it is enough to consider the problem on the half line.
\end{theorem}

\begin{definition}
\label{strong_m}The function $\varpi\in E$ is a strongly optimal solution to
$\mathbb{P}\left(  f\right)  $ if

$\left(  i\right)  $ For any bounded interval $\left[  T_{1},T_{2}\right]  ,$%
\[
\zeta_{T}\left(  f,\mathcal{V}_{\varpi}\left(  T_{1}\right)  ,\mathcal{V}%
_{\varpi}\left(  T_{2}\right)  \right)  =\frac{1}{T}\int_{0}^{T}f\left(
\bar{\varpi}\left(  t\right)  ,\bar{\varpi}^{\prime}\left(  t\right)
,\bar{\varpi}^{\prime\prime}\left(  t\right)  \right)  dt,
\]

$\left(  ii\right)  $
\[
\psi_{f}=\liminf\limits_{T\rightarrow\infty}J_{T}\left(  f,\varpi\right)  .
\]
the set of all strongly optimal solutions is denoted by $\Xi\left(  f\right)
,$ and that which are periodic is denoted by $\tilde{\Xi}\left(  f\right)  .$
\end{definition}

\cite{Lwz} showed $\Xi\left(  f\right)  $ and $\tilde{\Xi}\left(  f\right)  $
are all nonempty. If no confusion arises, all minimizers below would mean
\textit{strongly optimal solutions}.

\begin{theorem}
$\psi_{f}\left(  x\right)  $ is a convex function of $x.$
\end{theorem}

Let $x\in\mathbb{R}^{2},$ define
\[
\pi_{f}\left(  x\right)  =\inf\left\{  \left.  \liminf\limits_{T\rightarrow
\infty}\left[  I_{T}\left(  f;\varpi\right)  -T\psi_{f}\right]  \right\vert
\varpi\in W_{loc}^{2,1}\left(  \mathbb{R}^{+}\right)  ,\mathcal{V}_{\varpi
}\left(  0\right)  =x\right\}
\]

\begin{theorem}
For $T>0,$ $x,$ $y\in\mathbb{R}^{2}$. There are $\pi_{f}\left(  x\right)  $,
$\theta_{T}\left(  f,x,y\right)  $ such that $\zeta_{T}\left(  f,x,y\right)  $
may be decomposed as follows,%
\[
T\zeta_{T}\left(  f,x,y\right)  =T\psi_{f}+\pi_{f}\left(  x\right)  -\pi
_{f}\left(  y\right)  +\theta_{T}\left(  f,x,y\right)  ,
\]
where $\pi_{f}\left(  x\right)  $ is continuous w.r.t. $x$, $\theta_{T}\left(
f,x,y\right)  $ is a nonnegative continuous function of $\left(  T,x,y\right)
$, moreover,%
\[
\min_{y\in\mathbb{R}^{2}}\theta_{T}\left(  f,x,y\right)  =0,\text{ for any
}x\in\mathbb{R}^{2}.
\]

\end{theorem}

Here we list some properties that will be refered to in later proofs.

\begin{description}
\item[Property A] The mapping $T\longmapsto\left[  I_{T}\left(  f;\varpi
\right)  -T\psi_{f}\right]  $ is bounded on $\mathbb{R}^{+}$.
\end{description}

\begin{remark}
\label{good}It is a by-product from the proof of \cite{L&M}, that any funtion
$\varpi\in E$ satisfying definition \ref{strong_m} (i) possesses Property A.
This observation was later refined by \cite{M&Z1}.
\end{remark}

\begin{description}
\item[Property B] For any bounded interval $\left[  T_{1},T_{2}\right]  ,$
\[
I_{T_{1},T_{2}}\left(  f;\varpi\right)  -\left(  T_{2}-T_{1}\right)  \psi
_{f}+\pi_{f}\left(  \mathcal{V}_{\varpi}\left(  T_{1}\right)  \right)
-\pi_{f}\left(  \mathcal{V}_{\varpi}\left(  T_{2}\right)  \right)  .
\]

\end{description}

\begin{proposition}
Let $T>0,$ and $\varpi\in W^{2,1}\left(  \left[  0,T\right]  \right)  $. Then
$I_{T}\left(  f;\varpi\right)  <\infty$ if and only if $\varpi\in
W^{2,2}\left(  \left[  0,T\right]  \right)  .$

\begin{proof}
Assume $\varpi\in W^{2,2}\left(  \left[  0,T\right]  \right)  ,$ it easy to
see $I_{T}\left(  f;\varpi\right)  <\infty.$ On the other hand, $\varpi\in
W^{2,1}\left(  \left[  0,T\right]  \right)  $ and $I_{T}\left(  f;\varpi
\right)  <\infty$ imply $\varpi\in W^{2,2}\left(  \left[  0,T\right]  \right)
.$
\end{proof}
\end{proposition}

The above lemma tells us, the space $W^{2,1}\left(  \left[  0,T\right]
\right)  $ is enough for our problem though $W^{2,2}\left(  \left[
0,T\right]  \right)  $ seems more natural.

\begin{theorem}
[\cite{CMM}]The function $\psi_{f}\left(  x\right)  $ is convex w.r.t. $x$. In
particular, $\psi_{f}\left(  x\right)  $ is continuous.
\end{theorem}

The theorem has an important consequence which is fundamental to the proof of
the main result. We state it in the following lemma.

For the function $f\left(  x,y,z\right)  $ in $\left(  \ref{intgrd}\right)  $,
and $\lambda\in\mathbb{R}$, we define the lagrangian,%
\[
f_{\lambda}\left(  x,y,z\right)  =f\left(  x,y,z\right)  -\lambda x.
\]

\begin{lemma}
\label{convexity}For the function $\psi_{f}\left(  x\right)  $, the following
relations hold%
\begin{equation}
\psi_{f_{\lambda}}=\psi_{f}\left(  \eta\right)  -\lambda\eta,\forall\lambda
\in\partial\psi_{f}\left(  \eta\right)  ,\forall\eta\in\mathbb{R}.
\label{rel1}%
\end{equation}%
\begin{equation}
\Theta\left(  f;\eta\right)  \subset\Theta\left(  f_{\lambda}\right)
,\forall\lambda\in\partial\psi_{f}\left(  \eta\right)  ,\forall\eta
\in\mathbb{R}. \label{rel2}%
\end{equation}
Moreover, if $\eta$ is an exposed point of $\psi_{f}\left(  x\right)  $, then%
\begin{equation}
\Theta\left(  f;\eta\right)  =\Theta\left(  f_{\lambda}\right)  . \label{rel3}%
\end{equation}

\end{lemma}

\begin{proof}
Since $\psi_{f}\left(  x\right)  $ is convex, we may define its conjugate
function%
\[
\psi_{f}^{\ast}\left(  z\right)  =\sup_{x\in\mathbb{R}}\left[  xz-\psi
_{f}\left(  x\right)  \right]  .
\]
By remark \ref{exists}, we have%
\begin{align*}
&  -\psi_{f_{\lambda}}\\
&  =-\inf_{\varpi}\left\{  \lim_{T\rightarrow\infty}\frac{1}{T}\int_{0}%
^{T}f_{\lambda}\left(  \varpi,\varpi^{\prime},\varpi^{\prime\prime}\right)
dt\right\} \\
&  =-\inf_{\varpi}\left\{  \lim_{T\rightarrow\infty}\frac{1}{T}\int_{0}%
^{T}\left[  f\left(  \varpi,\varpi^{\prime},\varpi^{\prime\prime}\right)
-\lambda\varpi\right]  dt\right\} \\
&  =-\inf_{\varpi}\left\{  \lim_{T\rightarrow\infty}\frac{1}{T}\int_{0}%
^{T}f\left(  \varpi,\varpi^{\prime},\varpi^{\prime\prime}\right)
dt-\lambda\lim_{T\rightarrow\infty}\frac{1}{T}\int_{0}^{T}\varpi dt\right\} \\
&  =\sup_{\varpi}\left\{  \lambda\lim_{T\rightarrow\infty}\frac{1}{T}\int%
_{0}^{T}\varpi dt-\lim_{T\rightarrow\infty}\frac{1}{T}\int_{0}^{T}f\left(
\varpi,\varpi^{\prime},\varpi^{\prime\prime}\right)  dt\right\} \\
&  =\sup_{\xi\in\mathbb{R}}\sup_{\left[  \varpi\right]  =\xi}\left\{
\lim_{T\rightarrow\infty}\frac{1}{T}\int_{0}^{T}\left[  \lambda\xi-f\left(
\varpi,\varpi^{\prime},\varpi^{\prime\prime}\right)  \right]  dt\right\} \\
&  =\sup_{\xi\in\mathbb{R}}\sup_{\left[  \varpi\right]  =\xi}\left\{
\lambda\xi-\lim_{T\rightarrow\infty}\frac{1}{T}\int_{0}^{T}f\left(
\varpi,\varpi^{\prime},\varpi^{\prime\prime}\right)  dt\right\} \\
&  =\sup_{\xi\in\mathbb{R}}\left\{  \xi\lambda-\psi_{f}\left(  \xi\right)
\right\}  =\psi_{f}^{\ast}\left(  \lambda\right)  .
\end{align*}
hence%
\[
\psi_{f_{\lambda}}=-\psi_{f}^{\ast}\left(  \lambda\right)  =\inf_{\xi
\in\mathbb{R}}\left\{  \psi_{f}\left(  \xi\right)  -\xi\lambda\right\}  .
\]
The convexity of $\psi_{f}\left(  \xi\right)  $ implies, $\forall\eta
\in\mathbb{R},\forall\lambda\in\partial\psi_{f_{\lambda}}\left(  \eta\right)
,$%
\begin{equation}
\psi_{f}\left(  z\right)  -\lambda z\geqslant\psi_{f}\left(  \eta\right)
-\lambda\eta,\forall z\in\mathbb{R}, \label{conv_inequ}%
\end{equation}
that is, if $\lambda\in\partial\psi_{f_{\lambda}}\left(  \eta\right)  ,$ then%
\[
\inf_{\xi\in\mathbb{R}}\left\{  \psi_{f}\left(  \xi\right)  -\xi
\lambda\right\}  =\psi_{f}\left(  \eta\right)  -\lambda\eta.
\]
Thus $\left(  \ref{rel1}\right)  $ holds and,%
\[
\psi_{f_{\lambda}}=\psi_{f}\left(  \eta\right)  -\lambda\eta,\forall\lambda
\in\partial\psi_{f_{\lambda}}\left(  \eta\right)  ,\forall\eta\in\mathbb{R}.
\]

Suppose $\varpi\in\Theta\left(  f,\eta\right)  ,$%
\[
\psi_{f}\left(  \eta\right)  =\liminf\limits_{T\rightarrow\infty}\frac{1}%
{T}\int_{0}^{T}f\left(  \varpi,\varpi^{\prime},\varpi^{\prime\prime}\right)
dt.
\]
Noting $\left[  \varpi\right]  =\eta,$ we obtain%
\begin{equation}
\psi_{f}\left(  \eta\right)  -\lambda\eta=\liminf\limits_{T\rightarrow\infty
}\frac{1}{T}\int_{0}^{T}\left[  f\left(  \varpi,\varpi^{\prime},\varpi
^{\prime\prime}\right)  -\lambda\varpi\right]  dt, \label{conv_equ}%
\end{equation}
the left hand side is nothing else than $\psi_{f_{\lambda}}$, therefore
$\left(  \ref{rel2}\right)  $ is verified.

In addition, $\eta$ being an exposed point of $\psi_{f}\left(  x\right)  $
implies the equality in $\left(  \ref{conv_inequ}\right)  $ holds only when
$z=\eta$. Hence every $\varpi$ solving $\left(  \ref{conv_equ}\right)  $
should verify $\left[  \varpi\right]  =\eta,$ that is%
\[
\Theta\left(  f;\eta\right)  \subset\Theta\left(  f_{\lambda}\right)  ,
\]
showing that $\left(  \ref{rel3}\right)  $ is valid in the case of exposed point.
\end{proof}

\begin{remark}
The proof of the lemma also indicates%
\begin{equation}
\Theta\left(  f_{\lambda}\right)  =\bigcup\left\{  \Theta\left(
f;\eta\right)  ,\lambda\in\partial\psi_{f}\left(  \eta\right)  \right\}
,\forall\lambda\in\mathbb{R}. \label{rel_s}%
\end{equation}

\end{remark}

The following theorem is implied in \cite{L&M}, and explicitly stated in
\cite{M&Z1}, since the proof is not provided there, we give it here in detail.

\begin{theorem}
\label{bound_w1}Let $\lambda_{n}$ be a bounded sequence and $\tilde{\varpi
}_{n}$ the optimal periodic solution to problem $\mathbb{P}\left(
f_{\lambda_{n}}\right)  $. If the sequence of minimal energy $\psi
_{f_{\lambda_{n}}}$ is bounded. Then there exists a positive constant
$C>0,$such that%
\[
\left\Vert \tilde{\varpi}_{n}\right\Vert _{W^{1,\infty}\left(  \mathbb{R}%
\right)  }\leqslant C.
\]

\end{theorem}

\begin{proof}%
\begin{align*}
&  f_{\lambda_{n}}\left(  x,y,z\right) \\
&  =\frac{\xi^{2}}{2}z^{2}-\xi^{2}y^{2}-\lambda_{n}x+\frac{\tau-\xi^{2}}%
{2}x^{2}-\frac{\gamma}{6}x^{3}+\frac{1}{24}x^{4}\\
&  =\frac{\xi^{2}}{2}z^{2}-\xi^{2}y^{2}+\left(  x-\frac{\lambda_{n}}%
{2}\right)  ^{2}+\frac{\tau-\xi^{2}-2}{2}x^{2}-\frac{\gamma}{6}x^{3}+\frac
{1}{24}x^{4}-\frac{\lambda_{n}^{2}}{4}\\
&  \geqslant\frac{\xi^{2}}{2}z^{2}-\xi^{2}y^{2}+\frac{\tau-\xi^{2}-2}{2}%
x^{2}-\frac{\gamma}{6}x^{3}+\frac{1}{24}x^{4}-\frac{\lambda_{n}^{2}}{4}\\
&  \geqslant\frac{\xi^{2}}{2}z^{2}-\xi^{2}y^{2}+c_{1}x^{4}-c_{2},
\end{align*}
where $c_{1},$ $c_{2}$ are independent of $n$ and depends on $\xi,$ $\tau$ and
$\gamma$ only. Now, substituting $\varpi$ into these functionals, we have, for
$\forall\varpi\in W^{2,1}\left(  \left[  T_{1},T_{2}\right]  \right)
,W^{2,2}\left(  \left[  T_{1},T_{2}\right]  \right)  $
\begin{align*}
&  I_{T_{1},T_{2}}\left(  f_{\lambda_{n}},\varpi\right) \\
&  =\int_{T_{1}}^{T_{2}}f_{\lambda_{n}}\left(  \varpi,\varpi^{\prime}%
,\varpi^{\prime\prime}\right)  dt\\
&  \geqslant\int_{T_{1}}^{T_{2}}\left\{  \frac{\xi^{2}}{2}\left\vert
\varpi^{\prime\prime}\right\vert ^{2}-\xi^{2}\left\vert \varpi^{\prime
}\right\vert ^{2}+c_{1}\left\vert \varpi\right\vert ^{4}-c_{2}\right\}  ,
\end{align*}
an application of the interpolation inequality shows, there are positive
constants $\tilde{a}_{1}$ and $\tilde{a}_{2}$ such that%
\begin{align*}
&  I_{T_{1},T_{2}}\left(  f_{\lambda_{n}},\varpi\right) \\
&  \geqslant\int_{T_{1}}^{T_{2}}\left\{  \tilde{a}_{1}\left\vert
\varpi^{\prime\prime}\right\vert ^{2}+\tilde{a}_{2}\left\vert \varpi
\right\vert ^{4}-\tilde{c}_{2}\right\}  ,\text{ }\forall\varpi\in
W^{2,2}\left(  \left[  T_{1},T_{2}\right]  \right)  .\\
&  \geqslant\int_{T_{1}}^{T_{2}}\left\{  \tilde{a}_{1}\left\vert
\varpi^{\prime\prime}\right\vert ^{2}+\tilde{a}_{2}\left\vert \varpi
\right\vert ^{2}-\bar{c}_{2}\right\}  ,\text{ }\forall\varpi\in W^{2,2}\left(
\left[  T_{1},T_{2}\right]  \right)  .
\end{align*}
Note in the last inequality we employed $x^{\alpha}-x+1\geqslant0$ $\left(
\alpha>1,\forall x\geqslant0\right)  $. The Sobolev imbedding theoerm yields%
\begin{align*}
&  I_{T_{1},T_{2}}\left(  f_{\lambda_{n}},\varpi\right) \\
&  \geqslant\int_{T_{1}}^{T_{2}}\left\{  a_{1}\left\vert \varpi^{\prime
}\right\vert ^{2}+a_{2}\left\vert \varpi\right\vert ^{2}-a_{3}\right\}
,\text{ }\forall\varpi\in W^{2,1}\left(  \left[  T_{1},T_{2}\right]  \right)
,
\end{align*}
where $a_{1},$ $a_{2}$ and $a_{3}$ are positive constants, note they are
independent of $n$.%
\[
\lim_{T_{2}\rightarrow\infty}\frac{1}{T_{2}-T_{1}}I_{T_{1},T_{2}}\left(
f_{\lambda_{n}},\varpi\right)  =\min,
\]
then there is $T^{\prime}>T_{1}$ satisfying%
\[
I_{T^{\prime},T^{\prime}+1}\left(  f_{\lambda_{n}},\varpi\right)
\leqslant\min+1,
\]
the process can be proceeded to obtain a sequence $T_{n}\rightarrow\infty,$%
\[
I_{T_{n},T_{n}+1}\left(  f_{\lambda_{n}},\varpi\right)  \leqslant\min+1,
\]
by \cite[p171, Remark]{L&M} and \cite{Zas},%
\[
\left\vert \mathcal{V}_{\varpi}\left(  t\right)  \right\vert \leqslant M.
\]

\end{proof}

\begin{lemma}
[\cite{M&Z1}]\label{structure_per}Suppose that $\tilde{\varpi}$ is a periodic,
nontrivial minimizer for $\mathbb{P}\left(  f\right)  $, $\tau$ being its
minimal period. By an appropriate shift of variable, we may suppose%
\[
\tilde{\varpi}\left(  0\right)  =\min_{s\in\mathbb{R}}\tilde{\varpi}\left(
s\right)  ,
\]
then, there is $\tilde{\tau}\in\left(  0,\tau\right)  ,$ $\tilde{\varpi}$ is
strictly increasing on $\left[  0,\tilde{\tau}\right]  $ while strictly
decreasing on $\left[  \tilde{\tau},\tau\right]  $.
\end{lemma}

\begin{lemma}
[\cite{M&Z1}]\label{perfect_m}Suppose $\varpi\in E$ possessing Property A,
then for any $\varsigma\in\Omega\left(  \varpi\right)  ,$ there is
$\bar{\varpi}\in E$ possessing Property B such that%
\[
\left\{  \left.  \mathcal{V}_{\bar{\varpi}}\left(  s\right)  \right\vert
s\in\mathbb{R}\right\}  \subset\Omega\left(  \varpi\right)  ,V_{\bar{\varpi}%
}\left(  0\right)  =\varsigma,
\]
where%
\[
\Omega\left(  \varpi\right)  =\left\{  \left.  \nu\in\mathbb{R}^{2}\right\vert
\exists s_{n}\rightarrow\infty,\mathcal{V}_{\varpi}\left(  s_{n}\right)
\rightarrow\nu\right\}  .
\]

\end{lemma}

We are now in a position to prove the main theorem,

\section{Existence of minimizer}

This whole section is devoted to the proof of the main result.

Assume that $0$ is an extreme point of $\psi_{f}\left(  x\right)  $, then
there exists a sequence of exposed points $\left\{  \theta_{n}\right\}  $ of
$\psi_{f}\left(  x\right)  $ tending to $0$. Without loss of generality, we
may suppose that the sequence $\left\{  \theta_{n}\right\}  $ is
non-increasing and $\theta_{n}\neq0,\forall n$. Since each $\theta_{n}$ is an
exposed point, we have an optimal periodic solution $\tilde{\varpi}_{n}$ to
the constrained minimization problem $\mathbb{P}\left(  f;\theta_{n}\right)
$. Denote by $\tau_{n}$ the minimal positive period of $\tilde{\varpi}_{n}$.

The crux of the problem concentrates on the sequence of periods $\tau_{n}$. In
fact, if we show $\left\{  \tau_{n}\right\}  $ is bounded, then we can easily
derive a minimizer from the minimizing sequence $\left\{  \tilde{\varpi}%
_{n}\right\}  $. The condition $\left(  \ref{inequ_rcp}\right)  $ not only
guarantees the boundedness of $\left\{  \tau_{n}\right\}  $, but also ensures
us a nontrivial optimal solution. To prove the theorem, we employ an argument
of contradiction. That is, if $\left\{  \tau_{n}\right\}  $ is unbounded, then
we will reach the conclusion that, for $\forall\varepsilon>0,$ there is
$0\geqslant\varrho\geqslant-\varepsilon$ satisfying $\psi_{f}\left(
\varrho\right)  =f\left(  \varrho,0,0\right)  ,$ this obviously contradicts
condition $\left(  \ref{inequ_rcp}\right)  $. To make the it easier to follow,
we carry out the proof in several steps.

\begin{lemma}
The sequence $\left\{  \tilde{\varpi}_{n}\right\}  _{n\geqslant1}$ does not
admit a subsequence $\left\{  \tilde{\varpi}_{n_{k}}\right\}  _{k\geqslant1}$
that are all constant.
\end{lemma}

\begin{proof}
Suppose otherwise that $\left\{  \tilde{\varpi}_{n_{k}}\right\}  $ are all
constant functions. Since $\left[  \tilde{\varpi}_{n_{k}}\right]
=\theta_{n_{k}}$, then $\tilde{\varpi}_{n_{k}}\left(  t\right)  \equiv
\theta_{n_{k}}$. Each function $\tilde{\varpi}_{n_{k}}\left(  t\right)  $ is
an optimal solution of $\mathbb{P}\left(  f,\theta_{n_{k}}\right)  $, hence
\[
\psi_{f}\left(  \theta_{n_{k}}\right)  =\frac{1}{\tau_{n}}\int_{0}^{\tau_{n}%
}f\left(  \tilde{\varpi}_{n_{k}},\tilde{\varpi}_{n_{k}}^{\prime},\tilde
{\varpi}_{n_{k}}^{\prime\prime}\right)  dt=f\left(  \theta_{n_{k}},0,0\right)
.
\]
But the convex function $\psi_{f}\left(  x\right)  $ is continous at $x=0$,
therefore%
\[
\psi_{f}\left(  0\right)  =f\left(  0,0,0\right)  ,
\]
this contradicts $\left(  \ref{inequ_rcp}\right)  $. Thus we may assume all
$\left\{  \tilde{\varpi}_{n}\left(  t\right)  \right\}  _{n\geqslant1}$ are
not constant functions.

By lemma \ref{convexity}, there is $\lambda_{n}\in\partial\psi_{f}\left(
\theta_{n}\right)  $ such that $\tilde{\varpi}_{n}\in\mathbb{\tilde{M}}\left(
f_{\lambda_{n}}\right)  $ and $f_{\lambda_{n}}=f\left(  u,u^{\prime}%
,u^{\prime\prime}\right)  -\lambda_{n}u$. Since $\theta_{n}$ is non-increasing
and $\psi_{f}\left(  x\right)  $ is a convex function which is finite
everywhere, so $\left\{  \lambda_{n}\right\}  $ must be non-increasing and has
a lower bound, then $\left\{  \lambda_{n}\right\}  $ has a limit point
$\lambda^{\ast}$. By definition of subdifferential,
\[
\psi_{f}\left(  z\right)  \geqslant\psi_{f}\left(  \theta_{n}\right)
+\lambda_{n}\left(  z-\theta_{n}\right)  ,\forall z.
\]
therefore%
\[
\psi_{f}\left(  z\right)  \geqslant\psi_{f}\left(  0\right)  +\lambda^{\ast
}z,\forall z,
\]
namely, $\lambda^{\ast}\in\partial\psi_{f}\left(  0\right)  .$
\end{proof}

\begin{lemma}
\label{bd2&4}The sequence $\left\{  \tilde{\varpi}_{n}\right\}  $ is locally
bounded in $W^{4,2}\left(  \mathbb{R}\right)  ,$ that is, for any compact
interval $\left[  a,b\right]  $, there is a constant $C>0$ depending only on
$b-a,$ such that%
\begin{equation}
\left\Vert \tilde{\varpi}_{n}\right\Vert _{W^{4,2}\left(  \left[  a,b\right]
\right)  }\leqslant C,\forall n. \label{bd_ww2}%
\end{equation}

\end{lemma}

\begin{proof}
By theorem \ref{bound_w1}, there is a constant $C>0$, such that%
\begin{equation}
\left\Vert \tilde{\varpi}_{n}\right\Vert _{W^{1,\infty}\left(  \mathbb{R}%
\right)  }\leqslant C,\forall n. \label{bd_ww1}%
\end{equation}
Indeed, the function $f\left(  x,y,z\right)  $ allows lower bound of the form%
\[
f\left(  x,y,z\right)  \geqslant a_{1}\left\vert x\right\vert ^{4}%
-a_{2}\left\vert y\right\vert ^{2}+a_{3}\left\vert z\right\vert ^{2}-a_{4},
\]
where $a_{1},a_{2},a_{3},a_{4}\in\mathbb{R}$ are all positive constants.
Noting that $f_{\lambda_{n}}\left(  x,y,z\right)  =f\left(  x,y,z\right)
-\lambda_{n}x$ and $\lambda_{n}$ is bounded, then $f_{\lambda_{n}}$ also
admits the lower bound%
\[
f_{\lambda_{n}}\left(  x,y,z\right)  \geqslant\tilde{a}_{1}\left\vert
x\right\vert ^{4}-\tilde{a}_{2}\left\vert y\right\vert ^{2}+\tilde{a}%
_{3}\left\vert z\right\vert ^{2}-\tilde{a}_{4},
\]
where $\tilde{a}_{1},\tilde{a}_{2},\tilde{a}_{3},\tilde{a}_{4}\in\mathbb{R}$
are all positive constants independent of $n$. Therefore
\[
f_{\lambda_{n}}\left(  \tilde{\varpi}_{n},\tilde{\varpi}_{n}^{\prime}%
,\tilde{\varpi}_{n}^{\prime\prime}\right)  \geqslant\tilde{a}_{1}\left\vert
x\right\vert ^{4}-\tilde{a}_{2}\left\vert y\right\vert ^{2}+\tilde{a}%
_{3}\left\vert z\right\vert ^{2}-\tilde{a}_{4}.
\]
This inequality combining $\left(  \ref{bd_ww1}\right)  $ gives $\left(
\ref{bd_ww2}\right)  $. Furthermore, since $\tilde{\varpi}_{n}$ solves the E-L
equation%
\[
\tilde{\varpi}_{n}^{\prime\prime\prime\prime}+\tilde{\varpi}_{n}^{\prime
\prime}+\tilde{\varpi}_{n}+h^{\prime}\left(  \tilde{\varpi}_{n}\right)  =0.
\]
By $\left(  \ref{rel3}\right)  $ and the boundedness of $\lambda_{n}$,
$\tilde{\varpi}_{n}$ must be bounded in $W^{4,2}\left(  \left[  a,b\right]
\right)  ,$ namely, there is a constant $C>0$ ($C$ depends only on the length
$b-a$ of the interval)$,$ satisfying%
\begin{equation}
\left\Vert \tilde{\varpi}_{n}\right\Vert _{W^{4,2}\left(  \left[  a,b\right]
\right)  }\leqslant C,\forall n. \label{bd_ww4}%
\end{equation}

\end{proof}

\begin{lemma}
If $\tau_{n}\rightarrow\infty$, then $\forall\varepsilon>0,\exists
\tilde{\upsilon}_{\varepsilon}^{\ast}\in\tilde{\Xi}\left(  f_{\lambda^{\ast}%
}\right)  ,$ $\liminf\limits_{s\rightarrow\infty}\tilde{\upsilon}%
_{\varepsilon}^{\ast}\left(  s\right)  \geqslant-\varepsilon.$
\end{lemma}

\begin{proof}
Since the functional $I_{T_{1},T_{2}}\left(  f,\varpi\right)  $ is independent
of the time variable $t$, we may assume%
\[
\tilde{\varpi}_{n}\left(  0\right)  =\min_{s\in\mathbb{R}}\tilde{\varpi}%
_{n}\left(  s\right)  ,\forall n.
\]
Then $\tilde{\varpi}_{n}\left(  0\right)  <\theta_{n}$. By lemma
\ref{structure_per}, there is $\bar{\tau}_{n}\in\left(  0,\tau_{n}\right)  $,
$\tilde{\varpi}_{n}$ is strictly increasing on $\left(  0,\bar{\tau}%
_{n}\right)  $ and strictly decreasing on $\left(  \bar{\tau}_{n},\tau
_{n}\right)  $. Consider the set%
\[
\kappa_{n}=\left\{  \left.  s\in\left(  0,\tau_{n}\right)  \right\vert
\tilde{\varpi}_{n}\left(  s\right)  \geqslant\theta_{n}-\varepsilon\right\}
,
\]
by virtue of monotonic structure of $\tilde{\varpi}_{n}$, $\kappa_{n}$ may be
written as%
\[
\kappa_{n}=\left[  \check{\kappa}_{n},\bar{\tau}_{n}\right]  \cup\left[
\bar{\tau}_{n},\hat{\kappa}_{n}\right]  .
\]
Simple calculations show, ($\left\vert \kappa_{n}\right\vert $ denotes the
length of $\kappa_{n}$)
\begin{align*}
\theta_{n}  &  =\frac{1}{\tau_{n}}\int_{0}^{\tau_{n}}\tilde{\varpi}_{n}\left(
s\right)  ds\\
&  =\frac{1}{\tau_{n}}\left(  \int_{\kappa_{n}}\tilde{\varpi}_{n}\left(
s\right)  ds+\int_{\left.  \left(  0,\tau_{n}\right)  \right\backslash
\kappa_{n}}\tilde{\varpi}_{n}\left(  s\right)  ds\right) \\
&  \leqslant C\frac{\left\vert \kappa_{n}\right\vert }{\tau_{n}}+\left(
\theta_{n}-\varepsilon\right)  \left(  1-\frac{\left\vert \kappa
_{n}\right\vert }{\tau_{n}}\right)  ,
\end{align*}
this implies%
\[
\liminf\limits_{n\rightarrow\infty}\frac{\left\vert \kappa_{n}\right\vert
}{\tau_{n}}>0.
\]
Since otherwise%
\[
\liminf\limits_{n\rightarrow\infty}\frac{\left\vert \kappa_{n}\right\vert
}{\tau_{n}}=0
\]
would lead to the contradictory inequality $0\leqslant-\varepsilon$.

However, $\tau_{n}\rightarrow\infty$. Hence $\left\vert \kappa_{n}\right\vert
\rightarrow\infty$ and either the length of $\left[  \check{\kappa}_{n}%
,\bar{\tau}_{n}\right]  $ or $\left[  \bar{\tau}_{n},\hat{\kappa}_{n}\right]
$ tends to infinity as $n\rightarrow\infty$. Without no loss of generality,
suppose%
\begin{equation}
\bar{\tau}_{n}-\check{\kappa}_{n}\rightarrow\infty. \label{assumption}%
\end{equation}
Define%
\begin{equation}
\tilde{\upsilon}_{n}\left(  t\right)  =\tilde{\varpi}_{n}\left(
t+\check{\kappa}_{n}\right)  , \label{v_n}%
\end{equation}
Then $\left\{  \tilde{\upsilon}_{n}\right\}  $ is locally bounded in
$W^{4,2}\left(  \mathbb{R}\right)  $ by boundedness $\left(  \ref{bd_ww2}%
\right)  $ of $\left\{  \tilde{\varpi}_{n}\right\}  $. Furthermore, since
$\tilde{\varpi}_{n}\in\tilde{\Xi}\left(  f_{\lambda_{n}}\right)  $, we have,
for any bounded interval $\left[  t_{1},t_{2}\right]  ,$%
\[
J_{t_{2}-t_{1}}\left(  f_{\lambda_{n}},\tilde{\varpi}_{n}\right)
=\zeta_{t_{2}-t_{1}}\left(  f_{\lambda_{n}},\mathcal{V}_{\tilde{\varpi}_{n}%
}\left(  t_{1}\right)  ,\mathcal{V}_{\tilde{\varpi}_{n}}\left(  t_{2}\right)
\right)  .
\]
Whence $\tilde{\upsilon}_{n}\left(  t\right)  $ must satisfy the minimal
relation%
\begin{equation}
J_{t_{2}-t_{1}}\left(  f_{\lambda_{n}},\tilde{\upsilon}_{n}\right)
=\zeta_{t_{2}-t_{1}}\left(  f_{\lambda_{n}},\mathcal{V}_{\tilde{\upsilon}_{n}%
}\left(  t_{1}\right)  ,\mathcal{V}_{\tilde{\upsilon}_{n}}\left(
t_{2}\right)  \right)  . \label{optimal}%
\end{equation}
Since $\left\{  \tilde{\upsilon}_{n}\right\}  $ is bounded in $W^{4,2}\left(
\left[  t_{1},t_{2}\right]  \right)  $, we may suppose $\tilde{\upsilon}_{n}$
converges weakly in $W^{4,2}\left(  \left[  t_{1},t_{2}\right]  \right)  $ to
some element $\tilde{\upsilon}^{\ast}\in W^{4,2}\left(  \left[  t_{1}%
,t_{2}\right]  \right)  .$ Note the definition of $\tilde{\upsilon}_{n}$
relies on $\varepsilon$, to emphasize this dependence, we will write
$\tilde{\upsilon}^{\ast}$ as $\tilde{\upsilon}_{\varepsilon}^{\ast}$. By
Sobolev (compact) imbedding theorem,%
\begin{equation}
W^{4,2}\left(  \left[  t_{1},t_{2}\right]  \right)  \hookrightarrow
C^{3}\left(  \left[  t_{1},t_{2}\right]  \right)  . \label{imbd}%
\end{equation}
That is, $\tilde{\upsilon}_{n}\rightarrow\tilde{\upsilon}^{\ast}$ in
$C^{3}\left(  \left[  t_{1},t_{2}\right]  \right)  $. Sending $n$ in $\left(
\ref{optimal}\right)  $ to infinity, we obtain%
\[
J_{t_{2}-t_{1}}\left(  f_{\lambda^{\ast}};\tilde{\upsilon}^{\ast}\right)
=\zeta_{t_{2}-t_{1}}\left(  f_{\lambda^{\ast}},\mathcal{V}_{\tilde{\upsilon
}^{\ast}}\left(  t_{1}\right)  ,\mathcal{V}_{\tilde{\upsilon}^{\ast}}\left(
t_{2}\right)  \right)  ,
\]
therefore, proposition \ref{good} implies that $\tilde{\upsilon}_{\varepsilon
}^{\ast}\in\tilde{\Xi}\left(  f_{\lambda^{\ast}}\right)  .$

For any $s>0,$ we have $s\subset\left[  0,\bar{\tau}_{n}-\check{\kappa}%
_{n}\right]  $ as long as $n$ is large enough. Since $\tilde{\upsilon}_{n}$ is
strictly increasing on $\left[  0,\bar{\tau}_{n}-\check{\kappa}_{n}\right]  $,
one must have, by $\left(  \ref{imbd}\right)  $,
\[
\tilde{\upsilon}_{\varepsilon}^{\ast}\left(  s\right)  =\lim_{n\rightarrow
\infty}\tilde{\upsilon}_{n}\left(  s\right)  \geqslant\lim_{n\rightarrow
\infty}\left(  \theta_{n}-\varepsilon\right)  =-\varepsilon,
\]
hence%
\[
\liminf\limits_{s\rightarrow\infty}\tilde{\upsilon}_{\varepsilon}^{\ast
}\left(  s\right)  \geqslant-\varepsilon
\]
We may take this inequality one step further by showing that the limit
$\lim\limits_{s\rightarrow\infty}\tilde{\upsilon}_{\varepsilon}^{\ast}\left(
s\right)  $ actually exists and $\varrho\geqslant-\varepsilon.$

Note $\tilde{\upsilon}_{\varepsilon}^{\ast}\in\tilde{\Xi}\left(
f_{\lambda^{\ast}}\right)  $ implies it is bounded on the real line. Moreover,
$\tilde{\upsilon}_{\varepsilon}^{\ast}\left(  s\right)  =\lim_{n\rightarrow
\infty}\tilde{\upsilon}_{n}\left(  s\right)  $ is non-decreasing on $\left[
0,\bar{\tau}_{n}-\check{\kappa}_{n}\right]  $, so $\lim_{s\rightarrow\infty
}\tilde{\upsilon}_{\varepsilon}^{\ast}\left(  s\right)  $ exists and it is
denoted by $\varrho$,%
\begin{equation}
\varrho=\liminf\limits_{s\rightarrow\infty}\tilde{\upsilon}_{\varepsilon
}^{\ast}\left(  s\right)  \geqslant-\varepsilon. \label{limit}%
\end{equation}

\end{proof}

\begin{lemma}
The function $\tilde{\upsilon}_{\varepsilon}^{\ast}$ is differentiable and
$\lim\limits_{s\rightarrow+\infty}\tilde{\upsilon}_{\varepsilon}^{\ast\prime
}\left(  s\right)  =0.$
\end{lemma}

\begin{proof}
Take $\sigma\geqslant1,$define%
\[
\xi_{n}\left(  s\right)  =\tilde{\upsilon}_{\varepsilon}^{\ast}\left(
s+n\right)  ,s\in\left[  0,\sigma\right]  .
\]
Since $\tilde{\upsilon}_{\varepsilon}^{\ast}\in\tilde{\Xi}\left(  f\right)  $,
$V_{\xi_{n}}\left(  s\right)  $ is bounded on $\left[  0,\sigma\right]  $. By
the same reasoning as lemma \ref{bd2&4}, we obtain that $\xi_{n}\left(
s\right)  $ is bounded in $W^{4,2}\left(  \left[  0,\sigma\right]  \right)  $,
one may suppose $\xi_{n}\left(  s\right)  $ converges weakly in $W^{4,2}%
\left[  0,\sigma\right]  $ to $\xi\left(  s\right)  ,$ the convergence is also
valid in the sense of $C^{3}\left(  \left[  0,\sigma\right]  \right)  $ by
Sobolev imbedding theorem, therefore%
\[
\xi\left(  s\right)  =\lim_{n\rightarrow\infty}\xi_{n}\left(  s\right)
=\lim_{n\rightarrow\infty}\tilde{\upsilon}_{\varepsilon}^{\ast}\left(
s+n\right)  =\varrho,\forall s\in\left[  0,\sigma\right]  ,
\]
namely, $\xi\left(  s\right)  \equiv\varrho,$ $\forall s\in\left[
0,\sigma\right]  .$

Now it is easy to see that $\xi_{n}^{\prime}\left(  s\right)
\rightrightarrows0,$ $s\in\left[  0,\sigma\right]  .$ For otherwise, there
exist $\delta>0$ and $s_{n_{k}}\in\left[  0,\sigma\right]  ,$%
\begin{equation}
\xi_{n_{k}}^{\prime}\left(  s_{n_{k}}\right)  =\left\vert \xi_{n_{k}}^{\prime
}\left(  s_{n_{k}}\right)  \right\vert \geqslant\delta.
\label{greater_than_delta}%
\end{equation}
Assume with no loss that $s_{n_{k}}\rightarrow\tilde{s}\in\left[
0,\sigma\right]  .$ For any $0\leqslant s_{1}<s_{2}\leqslant\sigma$,%
\begin{align*}
&  \left\vert \xi_{n}^{\prime}\left(  s_{1}\right)  -\xi_{n}^{\prime}\left(
s_{2}\right)  \right\vert \\
&  =\left\vert \int_{s_{1}}^{s_{2}}\xi_{n}^{\prime\prime}\left(  t\right)
dt\right\vert \\
&  \leqslant\int_{s_{1}}^{s_{2}}\left\vert \xi_{n}^{\prime\prime}\left(
t\right)  \right\vert dt\\
&  \leqslant\left(  s_{2}-s_{1}\right)  ^{1-\frac{1}{\gamma}}\left(
\int_{s_{1}}^{s_{2}}\left\vert \xi_{n}^{\prime\prime}\left(  t\right)
\right\vert ^{2}dt\right)  ^{\frac{1}{2}},
\end{align*}
This indicates that $\xi_{n}^{\prime}\left(  s\right)  $ is compact in
$C\left(  \left[  0,\sigma\right]  \right)  $, therefore%
\[
\xi_{n_{k}}^{\prime}\left(  s_{n_{k}}\right)  \rightarrow\xi^{\prime}\left(
\tilde{s}\right)  ,
\]
but $\xi^{\prime}\left(  \tilde{s}\right)  =0,$ which contradicts $\left(
\ref{greater_than_delta}\right)  $.

Now, since $\xi_{n}^{\prime}\left(  s\right)  \rightrightarrows0,$
$s\in\left[  0,\sigma\right]  .$ Hence%
\[
\tilde{\upsilon}_{\varepsilon}^{\ast\prime}\left(  z\right)  =\xi
_{\sigma\left(  z\right)  }^{\prime}\left(  z-\sigma\left(  z\right)  \right)
\rightarrow0,\text{ }z\rightarrow+\infty,
\]
where%
\[
\sigma\left(  z\right)  =\left\lfloor \frac{z}{\sigma}\right\rfloor
=\max\left\{  n\in\mathbb{N}:n\sigma\leqslant z<\left(  n+1\right)
\sigma\right\}  .
\]

\end{proof}

\begin{lemma}
For $\varrho$ defined early, $\psi_{f}\left(  \varrho\right)  =f\left(
\varrho,0,0\right)  .$
\end{lemma}

\begin{proof}
By $\left(  \ref{limit}\right)  ,$%
\[
\lim_{T\rightarrow\infty}\frac{1}{T}\int_{0}^{T}\tilde{\upsilon}^{\ast}\left(
s\right)  ds=\lim_{T\rightarrow\infty}\tilde{\upsilon}^{\ast}\left(  T\right)
=\varrho,
\]
but $\tilde{\upsilon}^{\ast}$ is an optimal solution to $\mathbb{P}\left(
f_{\lambda^{\ast}},\varrho\right)  $, and%
\begin{equation}
\psi_{f_{\lambda^{\ast}}}=\psi_{f}\left(  \varrho\right)  -\lambda^{\ast
}\varrho. \label{mm_equ}%
\end{equation}
Hence%
\[
\psi_{f}\left(  \varrho\right)  =f\left(  \varrho,0,0\right)  .
\]
Indeed, since $\tilde{\upsilon}^{\ast}\in\Xi\left(  f_{\lambda^{\ast}}\right)
$, we deduce from remark $\left(  \ref{good}\right)  $ that, there exists a
constant $\tilde{C}>0$ independent of $T$ such that%
\[
\left\vert T\psi_{f_{\lambda^{\ast}}}-I_{T}\left(  f_{\lambda^{\ast}}%
;\tilde{\upsilon}^{\ast}\right)  \right\vert \leqslant\tilde{C},\text{
}\forall T.
\]
Lemma $\left(  \ref{perfect_m}\right)  $ then ensures an element
$\breve{\upsilon}^{\ast}\in E$ possessing Property B and%
\[
\left\{  \left.  \mathcal{V}_{\breve{\upsilon}^{\ast}}\left(  s\right)
\right\vert s\in\mathbb{R}\right\}  \subset\Omega\left(  \tilde{\upsilon
}^{\ast}\right)  =\left\{  \left(  \varrho,0\right)  \right\}  ,
\]
which implies $\breve{\upsilon}^{\ast}\equiv\varrho$ and%
\[
\psi_{f_{\lambda^{\ast}}}=f_{\lambda^{\ast}}\left(  0,0,0\right)  =f\left(
\varrho,0,0\right)  -\lambda^{\ast}\varrho,
\]
combining $\left(  \ref{mm_equ}\right)  $, we arrive at%
\[
\psi_{f}\left(  \varrho\right)  =f\left(  \varrho,0,0\right)  .
\]

\end{proof}

\begin{lemma}
$\varrho\leqslant0.$
\end{lemma}

\begin{proof}
Previous arguments show $\tilde{\upsilon}^{\ast}\in\Theta\left(
f,\varrho\right)  \cap\Theta\left(  f_{\lambda^{\ast}}\right)  ,$ it can be
deduced from $\left(  \ref{rel_s}\right)  $ that $\lambda^{\ast}\in
\partial\psi_{f}\left(  \varrho\right)  .$ Since $\lambda_{n}$ is
non-increasing and tends to $\lambda^{\ast}$, and $\lambda^{\ast}\in
\partial\psi_{f}\left(  0\right)  $, then we should have $\varrho\leqslant0.$
\end{proof}

\begin{lemma}
$\tau_{n}\nrightarrow0.$
\end{lemma}

\begin{proof}
By $\left(  \ref{bd_ww4}\right)  $ and Sobolev imbedding theorem, one may
suppose that%
\[
\tilde{\varpi}_{n}\left(  t\right)  \rightrightarrows\tilde{\varpi}\left(
t\right)  ,\tilde{\varpi}_{n}^{\prime}\left(  t\right)  \rightrightarrows
\tilde{\varpi}^{\prime}\left(  t\right)  ,\tilde{\varpi}_{n}^{\prime\prime
}\left(  t\right)  \rightrightarrows\tilde{\varpi}^{\prime\prime}\left(
t\right)  .
\]
Assume $\tau_{n}\rightarrow0,$ take any $t\in\mathbb{R},$ note that%
\[
\frac{1}{\tau_{n}}\int_{t}^{t+\tau_{n}}\tilde{\varpi}_{n}\left(  s\right)
=\theta_{n},
\]
an application of mean value theorem gives,%
\[
\tilde{\varpi}_{n}\left(  \varsigma_{n}\right)  =\theta_{n},\varsigma_{n}%
\in\left(  t,t+\tau_{n}\right)  ,
\]
sending $n\rightarrow\infty,$ we have%
\[
\tilde{\varpi}_{n}\left(  t\right)  \rightrightarrows\tilde{\varpi}\left(
t\right)  =0,\forall t\in\mathbb{R}.
\]
however, $\tilde{\varpi}_{n}$ is an optimal solution, hence%
\[
\psi_{f}\left(  \theta_{n}\right)  =\frac{1}{\tau_{n}}\int_{0}^{\tau_{n}%
}f\left(  \tilde{\varpi}_{n},\tilde{\varpi}_{n}^{\prime},\tilde{\varpi}%
_{n}^{\prime\prime}\right)  \rightarrow f\left(  0,0,0\right)  ,
\]
whence we obtain by sending $n\rightarrow\infty,$%
\[
\psi_{f}\left(  0\right)  =f\left(  0,0,0\right)  ,
\]
which contradicts $\left(  \ref{inequ_rcp}\right)  $.
\end{proof}

\begin{proof}
[Final proof of the main theorem]The proceeding lemmas show that, if $\tau
_{n}\rightarrow\infty$, then for any $\varepsilon>0,$ there is $-\varepsilon
\leqslant\varrho\leqslant0,$ such that $\psi_{f}\left(  \varrho\right)
=f\left(  \varrho,0,0\right)  ,$ this is an obvious contradiction to $\left(
\ref{inequ_rcp}\right)  $. Hence $\tau_{n}$ is bounded and has at least one
limit point $\tau^{\ast}\in\left(  0,\infty\right)  .$ Assume without loss of
generality $\tau_{n}\rightarrow\tau^{\ast}$. Since $\varpi_{n}$ is bounded in
$W_{loc}^{4,2}\left(  \mathbb{R}\right)  $, there is a subsequence%
\[
\varpi_{n_{k}}\rightharpoonup\varpi^{\ast}\text{ }(W_{loc}^{4,2}\left(
\mathbb{R}\right)  ),
\]
then $\varpi_{n_{k}}$ also converges uniformly on compacts, hence
$\varpi^{\ast}$ is a periodic function, denote its period by $\tau^{\ast}.$
Therefore $\varpi^{\ast}$ must be a solution to the minimization problem
$\mathbb{P}\left(  f_{\lambda^{\ast}}\right)  $ with%
\[
\int_{0}^{\tau^{\ast}}\varpi^{\ast}\left(  s\right)  ds=0,
\]
whence $\varpi^{\ast}\left(  s\right)  $ is also optimal to $\mathbb{P}\left(
f,0\right)  $, $\ $it is not trivial by $\left(  \ref{inequ_rcp}\right)  $.
\end{proof}

\begin{acknowledgement}
In the end, special thanks go to professor K.C.Chang, part of the proof is
based on his advice. Thank him for the helpful discussions.
\end{acknowledgement}


\end{document}